\documentclass[11pt]{article}%

\usepackage{amsthm,amsmath,amsfonts}
\usepackage{graphicx, verbatim}

\usepackage{caption}
\usepackage{subcaption}

\theoremstyle{plain}
\newtheorem{theorem}{Theorem}

\theoremstyle{definition}

\theoremstyle{remark}

\newcommand{\E}{\mathbf{E}}
\newcommand{\p}{\mathbf{P}}

\newcommand{\N}{\mathbb{N}}

\newcommand{\bX}{\mathbf{X}}

\newcommand{\bq}{\mathbf{q}}

\newcommand{\ind}{\mathbb{I}}

\begin{document}

\title{Branching model with state dependent offspring 
distribution for Chlamydia spread}

\author{
P\'eter Kevei\thanks{
Bolyai Institute, University of Szeged, 
Aradi v\'ertan\'uk tere 1, 6720 Szeged, Hungary; 
e-mail: \texttt{kevei@math.u-szeged.hu}} \and  
M\'at\'e Szalai\thanks{
Bolyai Institute, University of Szeged, 
Aradi v\'ertan\'uk tere 1, 6720 Szeged, Hungary; 
e-mail: \texttt{szalaim@math.u-szeged.hu}}
}

\maketitle

\begin{abstract}
\emph{Chlamydiae} are bacteria with an interesting unusual
developmental cycle. A single bacterium in its infectious form 
(elementary body, EB) enters
the host cell, where it converts into its dividing form (reticulate body, RB),
and divides by binary fission. Since only the EB form is infectious, before 
the host cell dies, RBs start to convert into EBs. 
After the host cell dies RBs do not survive. We model the population 
growth by a 2-type discrete-time branching process, where the probability 
of duplication depends on the state. Maximizing the EB production leads 
to a  stochastic optimization problem. Simulation study shows 
that our novel model is able to reproduce the main features of the 
development of the population.
\end{abstract}

\section{Introduction}

\emph{Chlamydiae} are obligate intracellular bacteria which 
have a unique two-stage developmental cycle, with two 
forms, the elementary body (EB) and the 
reticulate body (RB). The EB is the infectious form and it is not 
capable to multiply. After infecting the 
host cell, the EB differentiates to RB. 
RBs multiply in the host cell by binary fission. 
After some time RBs redifferentiate to EBs. The EBs are then released from
the host cell ready to infect new host cells. 

This unique life-cycle triggered a lot 
of mathematical work to model the growth of the population. Wilson \cite{Wilson}
worked out a deterministic model taking into account the infected and 
uninfected host cells and the extracellular \emph{Chlamydia} concentration.
Wan and Enciso \cite{WanEnciso} formulated a deterministic model
for the quantities of RBs and EBs, and solved an optimal control problem  
to maximize the quantity of EBs when the host cell dies.
The same problem in a stochastic framework was investigated by 
Enciso et al.~\cite{Enciso} and Lee et al.~\cite{Lee}. 

RB divide repeatedly by binary fission, which expands the RB population. 
Then after a period of no conversion, RBs convert into EBs. 
It was shown recently by Lee et al.~\cite{Lee} 
using 3D electron microscopy method and manual counting that 
this conversion occurs asynchronously, so that some RBs are
converting into EBs, while others continue to divide.
Mathematical models suggested up to now are unable to reproduce this asynchronous
conversion, since both in the deterministic differential equation model 
in \cite{WanEnciso} and in the stochastic model in \cite{Enciso} the optimal 
conversion strategy is the so-called `bang-bang' strategy, that is, 
up to some time the population duplicates, then converts to EBs with the maximal 
possible rate.

Branching processes are well-known tools to model cell proliferation, see 
the  monographs by Haccou et al.~\cite{Haccou}, Kimmel and Axelrod 
\cite{Kimmel}.
In \cite{Enciso}, a continuous
time Markov chain model was introduced with time-dependent transition rates, and the 
cell death was assumed to be independent of the population process. 
Bogdanov et al.~\cite{Bogdanov} used a discrete-time Galton--Watson process 
to model \emph{chlamydia} growth in the presence of antibiotics.
Here we use 
a discrete-time branching process model, where the probability of duplication 
and the time of the cell death depends on the state of the process.
Finding the optimal conversion strategy leads to a stochastic optimization problem,
a so-called \emph{discrete-time Markov control process}, see 
e.g.~Hern\'andez-Lerma and Lasserre \cite{Hernandez}.
The only input of the process is a \emph{death-rate function} $d(x,y)$, which
determines the probability that the host cell dies if there are $x$ RBs and 
$y$ EBs. Simulation study shows that with a simple 
death-rate function our model is able to capture the real behavior
described recently in \cite{Lee}.

The paper is organized as follows. In section \ref{sect:model} we 
describe our model. In sections \ref{sect:indep} and \ref{sect:dep}
we analyze the cases when the host cell's death time is independent of, or 
depends on the process. The latter case is biologically more relevant.
Section \ref{sect:sim} contains a simulation study.

\section{The theoretical model} \label{sect:model}

Consider a two-type discrete-time Galton--Watson branching process 
$\bX^\pi = (\bX_n^\pi)_n =$ $(X_n^{\pi}, Y_n^{\pi})_n$, $n \geq 0$, together 
with a sequence of probabilities $\pi = (p_n)_n$. 
We assume that $\pi$ is 
adapted to the natural filtration $(\mathcal{F}_n)_n$ generated by $\bX$,
i.e.~$\mathcal{F}_n = \sigma ( \bX_k^\pi, k \leq n)$. 
Initially $\bX_0^\pi = (1,0)$, and the process evolves as
\begin{equation*} %\label{eq:defX}
\begin{split}
& X_{n+1}^{\pi} = \sum_{i=1}^{X_n^{\pi}} \xi_{n,i}, \\
& Y_{n+1}^{\pi} = Y_n^{\pi} +\sum_{i=1}^{X_n^{\pi}} 
\left( 1 - \frac{\xi_{n,i}}{2} \right), \quad n\geq0, 
\end{split}
\end{equation*}
where $(\xi_n, \xi_{n,i}), n=1,2,\ldots, i=1,2,\ldots$ are 
conditionally independent random variables given $(p_n)_n$, 
for fix $n$ the variables $(\xi_n, \xi_{n,i}),i=1,2,\ldots$ are identically 
distributed, such that 
$\p(\xi_n=2 | p_n )=p_n$, $\p(\xi_n=0 | p_n )=1-p_n$. 

Here $X_n^{\pi}$ stands for the number of RBs and $Y_n^{\pi}$ for the number of 
EBs in generation $n$. In generation $n$ each RB duplicates with 
probability $p_n$ and converts into EB with probability $1-p_n$.
If $\xi_{n,i} = 2$, then
the $i$th RB in generation $n$ duplicates, while if $\xi_{n,i} = 0$ then
it converts to EB.
The process $\pi$, the sequence of duplication probabilities, 
is adapted to $(\mathcal{F}_n)_n$, which
intuitively means that based on the whole past of the process 
the population determines its duplication probabilities.
In what follows, we call the random process $\pi$ a \emph{strategy}.

For the conditional expectations we obtain
\begin{equation*} %\label{eq:cond-exp}
\begin{split}
\E [ ( X_{n+1}^\pi, Y_{n+1}^\pi ) | \mathcal{F}_n] 
& = 
(2 p_n X_n^\pi, Y_n^\pi + (1-p_n) X_n^\pi) \\
& = \bX_n^\pi
\begin{pmatrix}
2p_n & 1-p_n \\
0 & 1 
\end{pmatrix}.
\end{split}
\end{equation*}
If $p_n$ depends only on the actual state $(X_n, Y_n)$, then 
the process is Markovian.

The process ends at a random time $T \in \{1,2,\ldots\}$ when the infected 
host cell dies. The aim of the bacterial population is to produce as 
many EBs as possible, that is to maximize $\E ( Y_T^{\pi})$ over all 
possible strategies $(p_n)$.
Denoting by $\mathcal{P}$ the set of all strategies, 
a strategy $\bq$ is \emph{optimal}, if
\[
\sup_{\pi \in \mathcal{P}} \E (Y_T^{\pi}) = \E (Y_T^{\bq}).
\]
Note that we do not claim neither existence nor uniqueness, see the remark
after Theorem \ref{thm:indep}.

The cause of the host cell's death and the distribution of its time
is not yet well-understood. Experiments show the lysis times of 
different host cells varies between 48 and 72 hours post infection (hpi), 
see Elwell et al.~\cite{Elwell}. 
Here we consider two models. If $T$ is independent of the process than
we can calculate explicitly the optimal strategy, which turns out to be 
a deterministic `bang-bang' strategy. Depending on the distribution of 
$T$, the population 
doubles up to some deterministic time ($p_n = 1$), and then all the RBs convert 
to EBs immediately ($p_n = 0$). This phenomena is analogous to the findings in 
the continuous time setup in \cite{Enciso}, where independence of $T$
and $\bX^\pi$ was tacitly assumed. 
Therefore, this model cannot explain the asynchronous conversion.
In our second model we assume that the host cell dies at time $n$ with a 
certain probability depending on $\bX_n^\pi$, such that more bacteria imply 
higher death probability. In this more complex and more realistic model
we can determine the optimal strategy only numerically. We found that 
asynchronous conversion happens naturally. In simulations we 
obtained similar behavior as in real experiments in \cite{Lee}.

\section{Death time $T$ is independent of $\bX$} \label{sect:indep}

Assume that the host cell's death time $T$ is independent of the 
process $\bX^\pi$.
Introduce the notation $\pi_\ell = (1,1, \ldots, 1, 0, 0, \ldots)$,
where the first $\ell \geq 0$ components are 1.

\begin{theorem} \label{thm:indep}
Assume that $T \geq 1$ is bounded and it is independent of $\bX^\pi$.
Let $\ell$ be such that 
\begin{equation} \label{eq:T-ell}
2^{\ell} \p ( T > \ell ) = \sup_{k \geq 0} 2^{k} \p ( T > k).
\end{equation}
Then the optimal strategy is $\pi_\ell$ with optimal value
\[
\sup_{\pi \in \mathcal{P}} \E (Y_T^{\pi}) =  \sup_{k \geq 0} 2^{k} \p ( T > k).
\]
\end{theorem}

There are distributions such that $\lim_{k \to \infty} 2^k \p( T > k) = \infty$,
in which case it is easy to see that there is no optimal strategy.
Furthermore, one can construct distributions for which $\ell$ in
\eqref{eq:T-ell} is not unique, showing that the optimal strategy is not 
necessarily unique.

\begin{proof}
To ease notation we suppress $\pi$.
Since $T \leq N$, for some $N$, in the optimal strategy $p_{N-1} = 0$. Next,
using the independence of $T$ and $X$
\[
\begin{split}
& \E \left[ Y_T | T > N - 2, \mathcal{F}_{N-2} \right] \\
& = 
Y_{N-2} + 2 p_{N-2} X_{N-2} \p ( T = N | T > N - 2) +
X_{N- 2 } ( 1 - p_{N - 2} ) \\
& =
Y_{N-2} +  X_{N-2} \left( 1 + p_{N-2} ( 2 \p ( T = N | T > N - 2) -1) \right).
\end{split}
\]
Thus, $p_{N-2}$ is either 0 or 1, and its value 
only depends on the distribution of $T$.
Iteration gives that the optimal strategy is deterministic, and each 
$p_i$ is 0 or 1.

This means that the population doubles up to generation $k$, then all the RBs
convert to EBs. These strategies are easy to compare. Under $\pi_k$
simply $Y_T = \ind ( T \geq k+1) 2^{k}$, with $\ind$ standing for the 
indicator function, thus 
\[
\E (Y_T) = \p ( T > k) 2^{k}.
\]
Taking the maximum in $k$, 
we obtain that $\pi_\ell$ is indeed the optimal strategy.
\end{proof}

\section{Death time $T$ depends on $\bX$} \label{sect:dep}

Here we assume that $T$, the death time depends on the process $\bX^\pi$.
Given that the host cell is alive in generation $n-1$, the probability that 
it dies in the next step is $d(X_n^\pi,Y_n^\pi)$, that is 
\[
\p ( T = n | T> n-1, \mathcal{F}_n ) = d(X_n^\pi,Y_n^\pi).
\]
The deterministic \emph{death-rate} function $d$ describes the effect of
RBs and EBs to cell's death. It is not clear which type is more harmful to 
the host cell, since RB particles are larger, while EB particles secrete chemicals
poisoning the host cell, see e.g.~\cite{Elwell}.
Assume that 
\begin{equation} \label{eq:d-ass}
\exists \ C > 0 \ \text{ such that }
\ d(x,y) = 1 \ \text{ whenever } \ x + y \geq C. 
\end{equation}
That is, if the total number of bacteria exceeds $C$ the host 
cell necessarily dies. This is biologically a natural assumption.

In this scenario the process is a special \emph{discrete-time Markov 
control process} (or Markov decision process). For theory and properties 
of these processes we refer to the monograph by 
Hern\'andez-Lerma and Lasserre \cite{Hernandez}.
To see that our model fits in the theory we slightly modify our process. 
Recall that $X_n = X_n^\pi$ depends on the strategy $\pi$, however 
for notational ease we suppress the upper index.
Let $\widetilde X_n = X_n \ind( T > n)$, 
$\widetilde Y_n = Y_n \ind (T \geq n)$.
Note that $\widetilde X_T = 0$, $\widetilde Y_T = Y_T$, and 
$\widetilde Y_{T+1} = 0$, which is convenient at the definition
of the reward function in \eqref{eq:reward}.
Then the state space is $\N \times \N$, the control set,  
the set of possible duplication probabilities, is $[0,1]$ for any 
state, and the transition probabilities are, for $x > 0$, $y \geq 0$,
\begin{equation} \label{eq:trans-prob}
\begin{split}
& \p \left( \widetilde X_{n+1} = 2 j, \widetilde  Y_{n+1} = y + x - j | 
\widetilde  X_n = x, \widetilde  Y_n = y, p_n = p \right) \\
& = \binom{x}{j} p^j (1-p)^{x-j} (1 - d(2j, y + x -j)), \quad j= 1,\ldots, x, \\
& \p \left( \widetilde  X_{n+1} = 0, \widetilde  Y_{n+1} = y + x - j | 
\widetilde  X_n = x, \widetilde  Y_n = y,  p_n = p \right) \\
& = \binom{x}{j} p^j (1-p)^{x-j} d(2j, y + x -j), \quad j= 1,\ldots, x, \\
& \p \left( \widetilde  X_{n+1} = 0, \widetilde  Y_{n+1} = y + x  |
\widetilde  X_n = x, \widetilde  Y_n = y,  p_n = p \right) \\
& = (1-p)^{x}, 
\end{split}
\end{equation}
while if $x = 0$
\[
\p( \widetilde X_{n+1} = 0, \widetilde Y_{n+1} = 0 |
\widetilde X_n = 0,  \widetilde Y_n = y, p_n = p ) = 1.
\]
The first two formulae correspond to the possibility that $j \geq 1$ bacteria 
duplicate (with probability $\binom{x}{j} p^{j} (1-p)^{x-j}$)
and the host cell remains alive, or die, while the third formula 
corresponds to the possibility that all the RBs convert to EBs, and in this 
case it does not matter whether the host cell dies or not. The fourth equation states
that $(0,0)$ is the unique absorbing state, which is a convenient condition for 
the form of the reward function.

The \emph{reward function} ($-1$ times the cost function in \cite{Hernandez})
gives the number of EBs upon cell's death, that is
\begin{equation} \label{eq:reward}
c(x,y) =
\begin{cases}
y, & x = 0, \\
0, & \text{otherwise}
\end{cases}
\end{equation}
Define the \emph{value function} 
\begin{equation} \label{eq:h}
h(x,y ) = 
\begin{cases}
\sup_{\pi \in \mathcal{P}} 
\E \left[ \sum_{n=0}^\infty c(\widetilde X_n, \widetilde Y_n)  |  
(\widetilde X_0, \widetilde Y_0 ) = (x, y) \right], 
& d(x,y) < 1, \\
y, & d(x,y) = 1.
\end{cases}
\end{equation}
which is the optimal number of expected EBs upon host cell's death, given 
that the host cell is alive and $(\widetilde X_0 , \widetilde Y_0) = (x,y)$, 
if $d(x,y) < 1$. 
If $d(x,y) = 1$ then the cell cannot be alive at state $(x,y)$, 
thus the reward is $y$. Clearly $h(0,y) = y$. 
Note that since $(0,0)$ is the only absorbing state,
in the infinite sum in \eqref{eq:h} there is only one non-zero term.

We are looking for the value $h(x,y)$ and the optimal strategy $\pi$.
This stochastic optimization problem is in fact a finite-horizon problem, 
see \cite[Chapter 3]{Hernandez}. Indeed, from any state 
$(\widetilde X_n, \widetilde Y_n ) = (x,y)$ either the total number of 
bacteria increases ($j \geq 1$ and the host cell survives in \eqref{eq:trans-prob}),
or $\widetilde X_{n+1} = 0$, meaning that the cell dies. Therefore, 
by condition \eqref{eq:d-ass} from any initial state $(x,y)$ the process 
reaches the absorbing state $(0,0)$ in at most $C +1$ steps. So in \eqref{eq:h}
in the summation the upper limit can be changed to $C$.
Using Theorem 3.2.1 in \cite{Hernandez} both the value function $h$ 
and the optimal strategy can be determined by backward induction
on time. In our setup, backward induction on the total number of bacteria is 
more natural, and this goes as follows.

\begin{theorem} \label{thm:backward}
Assume that \eqref{eq:d-ass} holds.  Then $h(x,y) = y$ if $x + y \geq C$, and 
$h(0,y) = y$ for any $y$. Assume that $h(x,y)$ is determined whenever $x+y \geq m$ for 
some $m \leq C$, and let $x + y = m-1$. Then
\begin{equation} \label{eq:h-recursion}
\begin{split}
& h(x,y)  = \max_{p \in [0,1]} 
\sum_{j=0}^x \binom{x}{j} p^j (1-p)^{x-j} \\
& \times 
\left[ 
d(2j, y+x-j) (y+x - j) + (1-d(2j, y+x - j)) h(2j, y+x - j) 
\right],
\end{split}
\end{equation}
where all the values of $h$ on the right-hand side are determined.
The maximum in $p$ of the continuous function on the right-hand side 
of \eqref{eq:h-recursion} is attained at $p(x,y)$,
which gives the optimal strategy.
\end{theorem}

\begin{proof}
From definition \eqref{eq:h} we see that $h(x,y) = y$ if $x +y \geq C$ or
$x  = 0$. Formula \eqref{eq:h-recursion} follows from the Markovian structure
and from the transition probabilities in \eqref{eq:trans-prob}.

Indeed, from $(x,y)$, $x+y = m-1$, 
the possible states in \eqref{eq:h-recursion} are $(2j, y+x-j)$, $j=0,1,\ldots, x$,
with total number of bacteria $x+y+j$, therefore $h(2j,y+x-j)$ is determined
for $j \geq 1$ by the induction assumption, and for $j= 0$ by definition
$h(0,x+y) = x+y$. Thus all the quantities in \eqref{eq:h-recursion} are known,
so $h(x,y)$ can be calculated.
\end{proof}

\section{Simulation studies} \label{sect:sim}

For a given death-rate function $d$, we can determine 
numerically the value function and the optimal strategy using 
Theorem \ref{thm:backward}. Then, the process is a simple 
Galton--Watson branching process with state-dependent 
offspring distribution, which can be simulated easily.
In each examples below 
the empirical mean of RBs and EBs are calculated from 1000 simulations.

\smallskip

First we consider a simple threshold death-rate function, that is 
for some $C > 0$
\begin{equation} \label{eq:death0}
d_1(x,y) = d_1(x+y) = 
\begin{cases}
1, & \text{ if } x + y \geq C, \\
0, & \text{ otherwise.}
\end{cases}
\end{equation}
This is a simple, but biologically very unnatural death-rate function.
In this case the bacteria typically start to convert to EBs at an early 
stage, and after a lot of generation it reaches the optimal bound $C-1$.

\begin{figure} 
\begin{center}
\includegraphics[width = 0.9\textwidth]{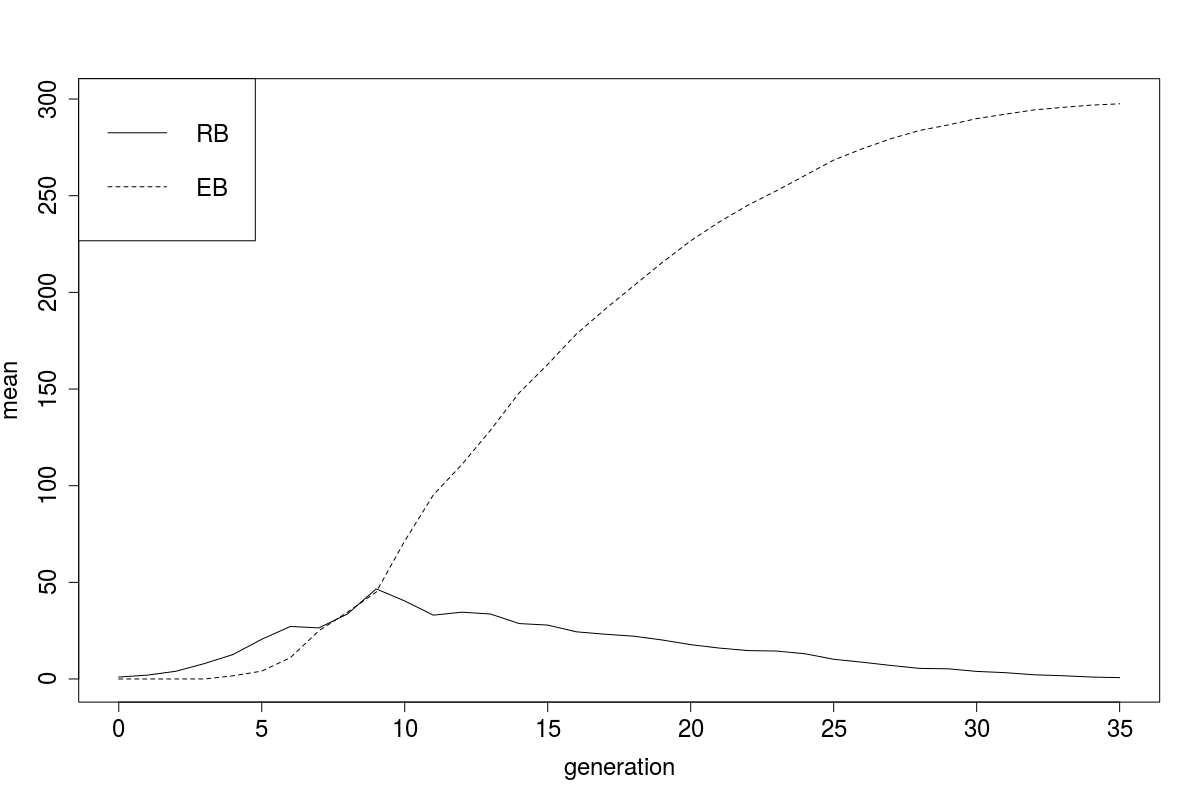}
\end{center}
\vspace*{-0.7cm}
\caption{The mean number of EBs and RBs for the death-rate \eqref{eq:death0} with
$C = 300$.}
\label{fig:mean-d0}
\end{figure}

For simulations we choose $C = 300$. The  value function is almost 
constant 299 with $h(1,0) = 298.7$. In 
Figure \ref{fig:mean-d0} we see that 
the number of RBs is typically small, while the number of EBs starts 
to increase at an early stage. In Figure \ref{fig:sim-d0} 
there are six trajectories of the process. On Figure \ref{fig:p} 
(top left) we see the numerical $p$ values. The structure of the 
death-rate function causes the discontinuity of the $p$ function. Note 
e.g.~that $p(x,y) \equiv 1$ on the line $\{ (x,y) : 2 x + y = 299 \}$, since 
after one duplication the population reaches the maximum possible value $299$.

\begin{figure} 
\begin{center}
\includegraphics[width = 0.9\textwidth]{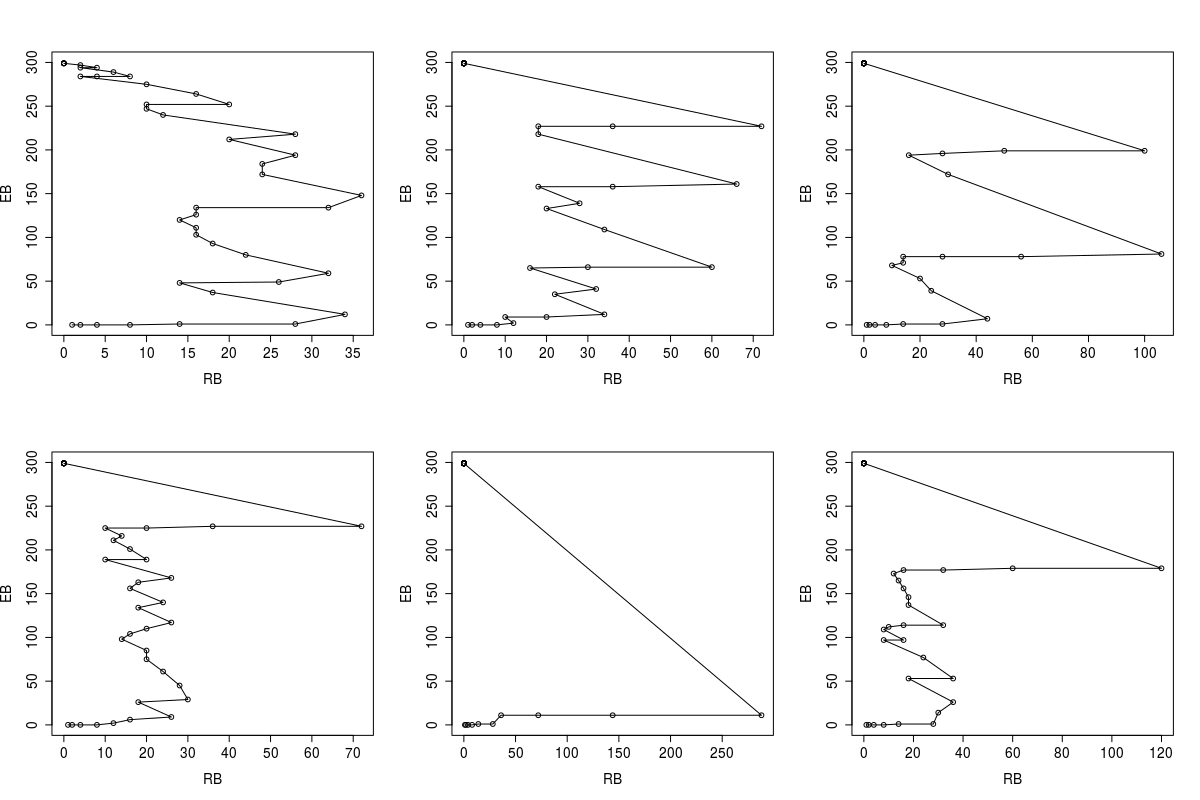}
\end{center}
\vspace*{-0.7cm}
\caption{Simulations of the process with death-rate \eqref{eq:death0} and
$C = 300$.}
\label{fig:sim-d0}
\end{figure}

Consider a smoother death-rate function 
\begin{equation} \label{eq:death2}
d_2(x,y) = 
\begin{cases}
1 - e^{-c_0 ( \alpha x + \beta y)}, & x + y \leq C -1, \\
1, & x + y \geq C.
\end{cases}
\end{equation}
If the total number of bacteria is small, then
it is unlikely that the host cell dies. The parameters $\alpha, \beta$
allows us to tune the relative effect of EBs and RBs on the host cell's death.
On the one hand RBs are much larger than EBs suggesting $\alpha > \beta$,
on the other hand EBs secrete chemicals enhancing cell death.
Note that biological experiments suggests that 
\emph{chlamydia} controls host cell survival, see \cite[p.~392]{Elwell}.
We tried three scenarios, with $c_0 = 0.0003$ in each cases and  
$(\alpha, \beta) = (1,3)$, $(2,2)$, and $(3, 1)$, with 
$C = 2500$, $1500$, $3000$, respectively. 
We chose $C$ large enough, so that the optimal strategy does not depend on
its actual value. The rationale of choice of the 
different threshold values $C$ can be seen from Figure \ref{fig:p}.
For the empirical mean of 1000 simulations and some typical trajectories see 
Figures \ref{fig:mean-d1-1} -- \ref{fig:sim-d1-3}.

\begin{figure} 
\begin{center}
\includegraphics[width = 0.9\textwidth]{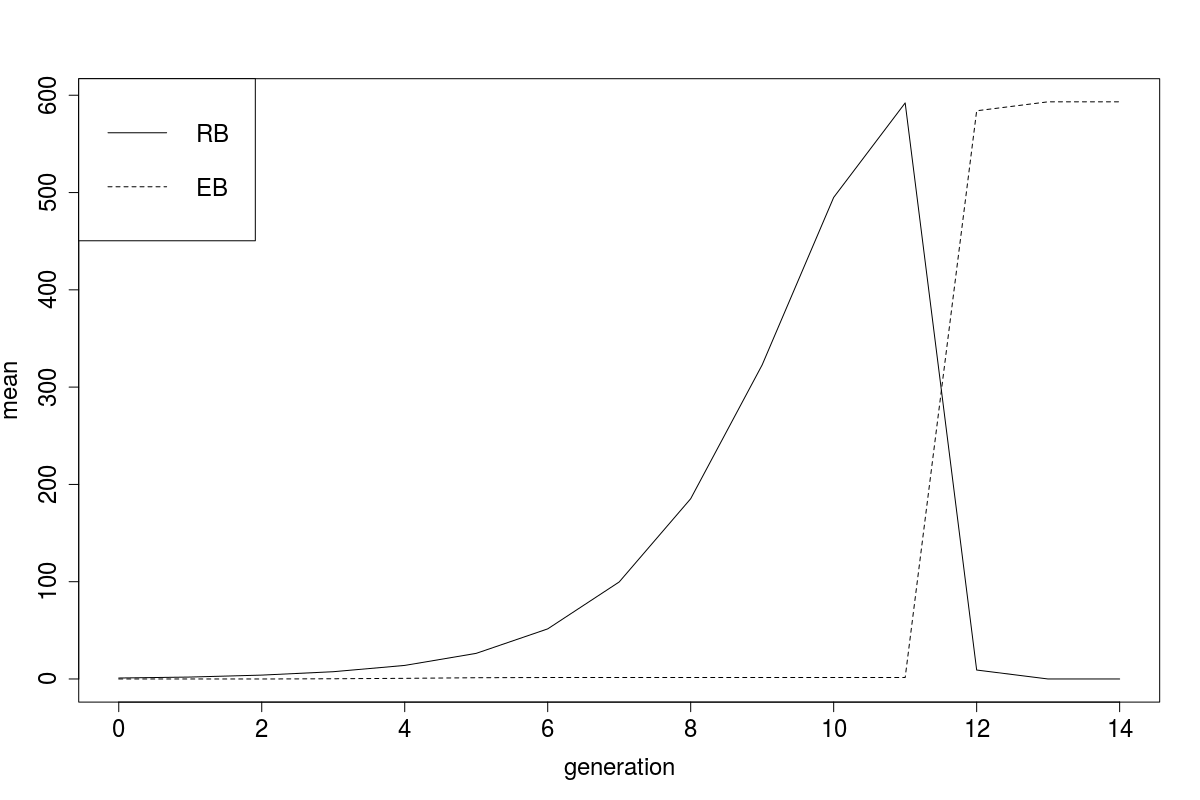}
\end{center}
\vspace*{-0.7cm}
\caption{The mean number of EBs and RBs for the death-rate \eqref{eq:death2},
with $(\alpha, \beta) = (1,3)$.}
\label{fig:mean-d1-1}
\end{figure}

\begin{figure} 
\begin{center}
\includegraphics[width = 0.9 \textwidth]{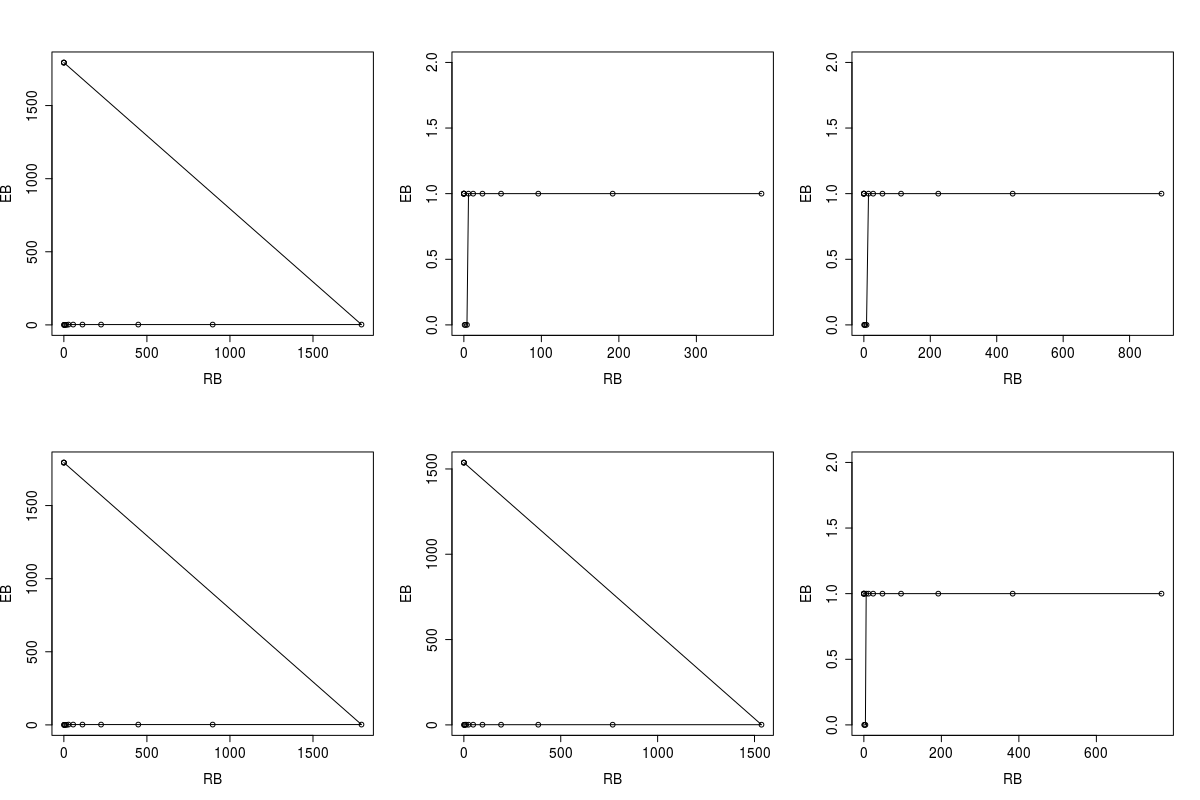}
\end{center}
\vspace*{-0.7cm}
\caption{Simulations of the process with death-rate \eqref{eq:death2},
with $(\alpha, \beta) = (1,3)$.}
\label{fig:sim-d1-1}
\end{figure}

A quick look on the mean functions
and on some trajectories of the processes 
reveals that the population behaves 
very differently. 
For $(\alpha, \beta) = (1,3)$ the relative effect of EBs on cell-death
is much larger. Therefore, the process prefer to have only RBs up to 
some point (generation 11), and then all RBs convert to EBs immediately, resulting
an `all or nothing' strategy.
The exponential increase of RBs and the sudden change is clearly visible 
both on the means (Figure \ref{fig:mean-d1-1}), and on the trajectories 
(Figure \ref{fig:sim-d1-1}). 
Here $h(1,0) = 605$.
The optimal $p$ values on Figure \ref{fig:p}
(top right) show the same pattern: in each state 
either all cells duplicate ($p=1$), or all cells convert ($p= 0$).

\begin{figure} 
\begin{center}
\includegraphics[width = 0.9\textwidth]{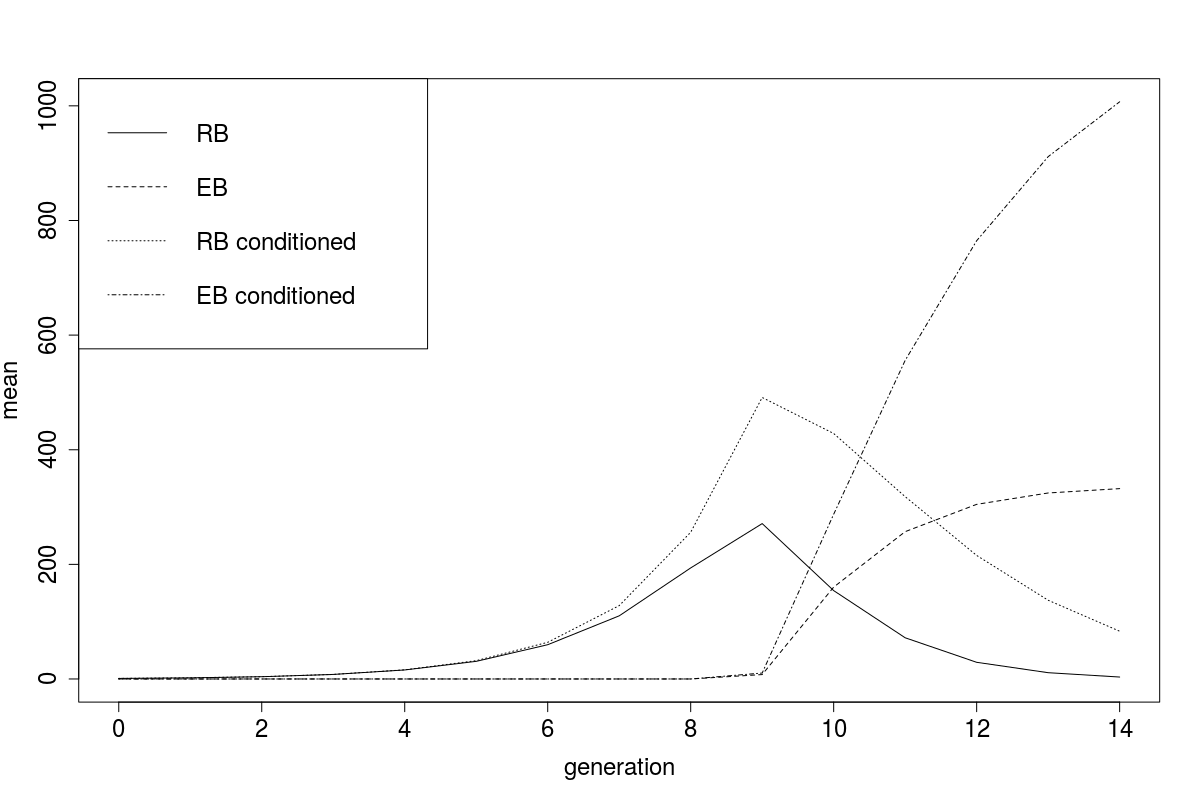}
\end{center}
\vspace*{-0.7cm}
\caption{The mean and conditional mean of EBs and RBs for the death-rate 
\eqref{eq:death2}, with $(\alpha, \beta) = (2,2)$.}
\label{fig:mean-d1-2}
\end{figure}

\begin{figure} 
\begin{center}
\includegraphics[width = 0.9\textwidth]{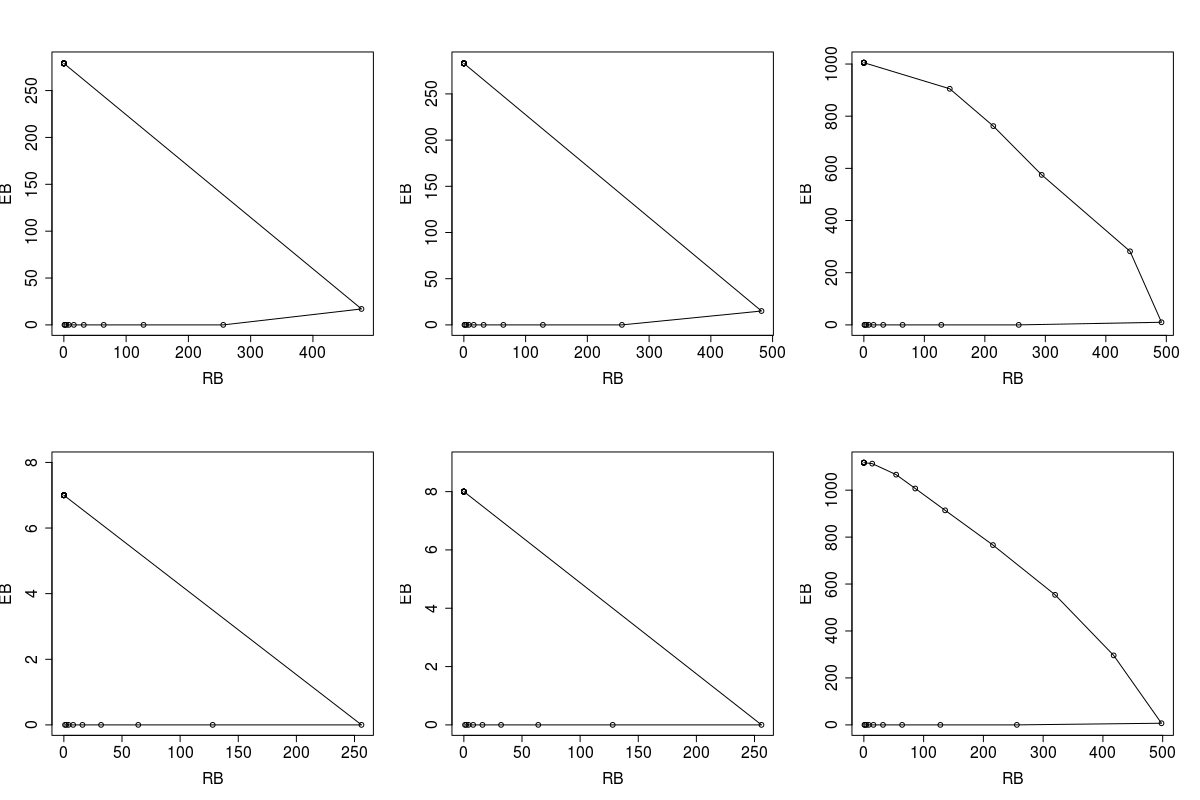}
\end{center}
\vspace*{-0.7cm}
\caption{Simulations of the process with death-rate \eqref{eq:death2},
with $(\alpha, \beta) = (2,2)$.}
\label{fig:sim-d1-2}
\end{figure}

For $(\alpha, \beta) = (2,2)$ the relative 
effect of RBs and EBs is the same. The RBs duplicate and increase exponentially
up to generation 9, then they start to convert to EBs.
In Figures \ref{fig:mean-d1-2} and \ref{fig:sim-d1-2} 
we see that in generations 9--12 the EBs and RBs simultaneously appear,
showing the asynchronous conversion obtained in real experiments in 
\cite{Lee}. Here $h(1,0) = 324$.

\begin{table}[h!]
\centering
\begin{tabular}{c||cccccccc}
% Generation 
%gen / hpi
%& 1 / 12  & 3 / 16  & 5 / 20  & 8/ 24 
%& 11 / 28  & 12 / 32  & 13 / 36 & 14 / 40 \\ \hline
gen & 0   & 3   & 5  & 7 & 10  & 11 & 12  & 13  \\ \hline
hpi & 12  & 16  & 20  & 24 & 28  & 32 & 36 & 40 \\ \hline
RB  & 1  &  8 & 32 & 128 & 430 & 318 & 213 & 142 \\
EB  & 0  & 0  & 0  & 0 & 287 & 559  & 770  & 910    \\
\hline \hline 
RB  & 1.3 & 7.6 & 34 & 105 & 385 & 507 & 271 & 171 \\
EB  & 0   & 0   & 0  & 3.7 & 192 & 656 & 706 & 751 \\
\end{tabular}
\caption{Conditional means for $(\alpha, \beta) = (2,2)$ in our simulations,
and real data from Lee et al.~\cite{Lee} (last two rows)}
\end{table} \label{tab:1}

In Figure \ref{fig:mean-d1-2}
we also plotted the empirical means of the EBs and RBs conditioned on 
that the host cell is alive. 
In the real experiment in Lee et al.~\cite{Lee} only those 
inclusions are counted where the host cell is alive. This 
clearly causes a bias. We can transform the generation time to 
real time, hours-post-infection (hpi). After the EB enters the host cell,
it takes about 12 hours to convert to RB and to start to duplicate.
Between 12 and 24 hpi the doubling time is about 1.8 hours, and 
around 28 hpi RB-EB conversion starts, see \cite[p.~2]{Lee}.
In Table \ref{tab:1} we copied the measurements from \cite{Lee} to 
see that our model captures extremely well the experimental data.

For the optimal $p$ values 
in Figure \ref{fig:p} (bottom left) we do see values other than 0 and 1. 
There are
no big jumps in the $p$ values, which makes it biologically relevant.

\begin{figure} 
\begin{center}
\includegraphics[width = 0.9\textwidth]{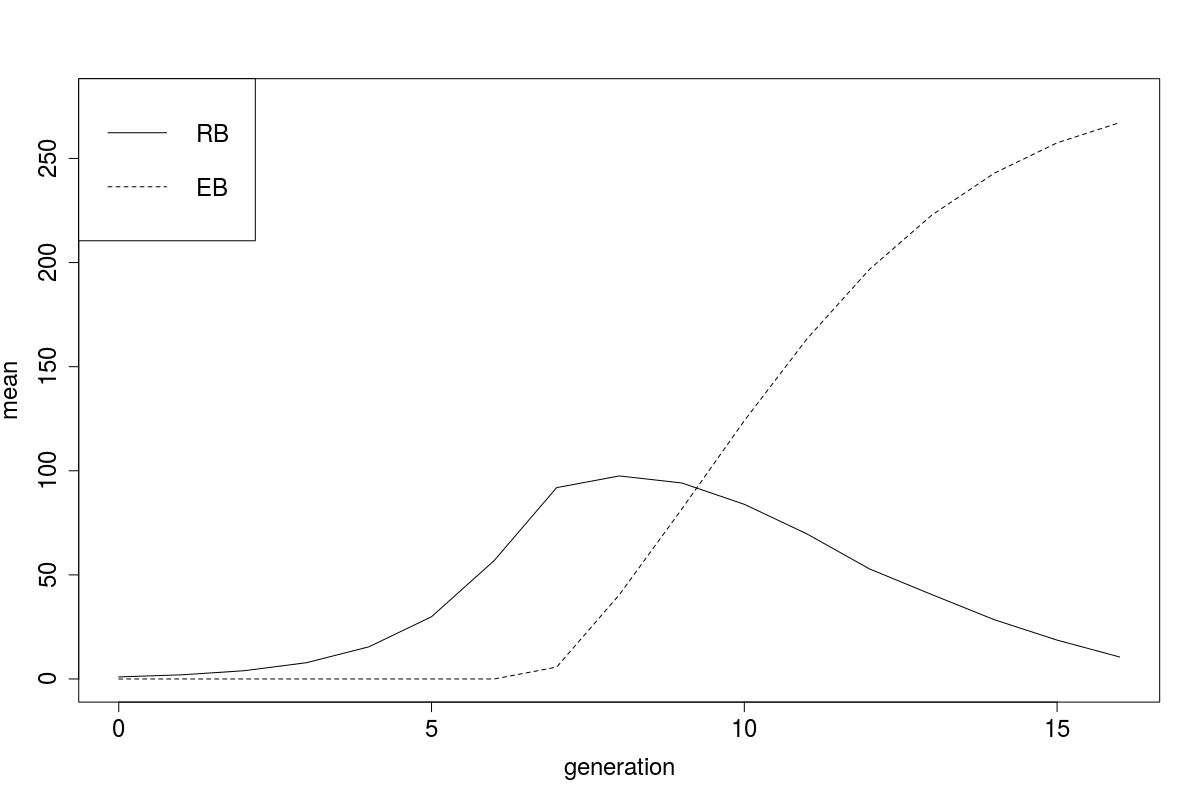}
\end{center}
\vspace*{-0.7cm}
\caption{The mean number of EBs and RBs for the death-rate \eqref{eq:death2},
with $(\alpha, \beta) = (3,1)$.}
\label{fig:mean-d1-3}
\end{figure}

\begin{figure} 
\begin{center}
\includegraphics[width = 0.9\textwidth]{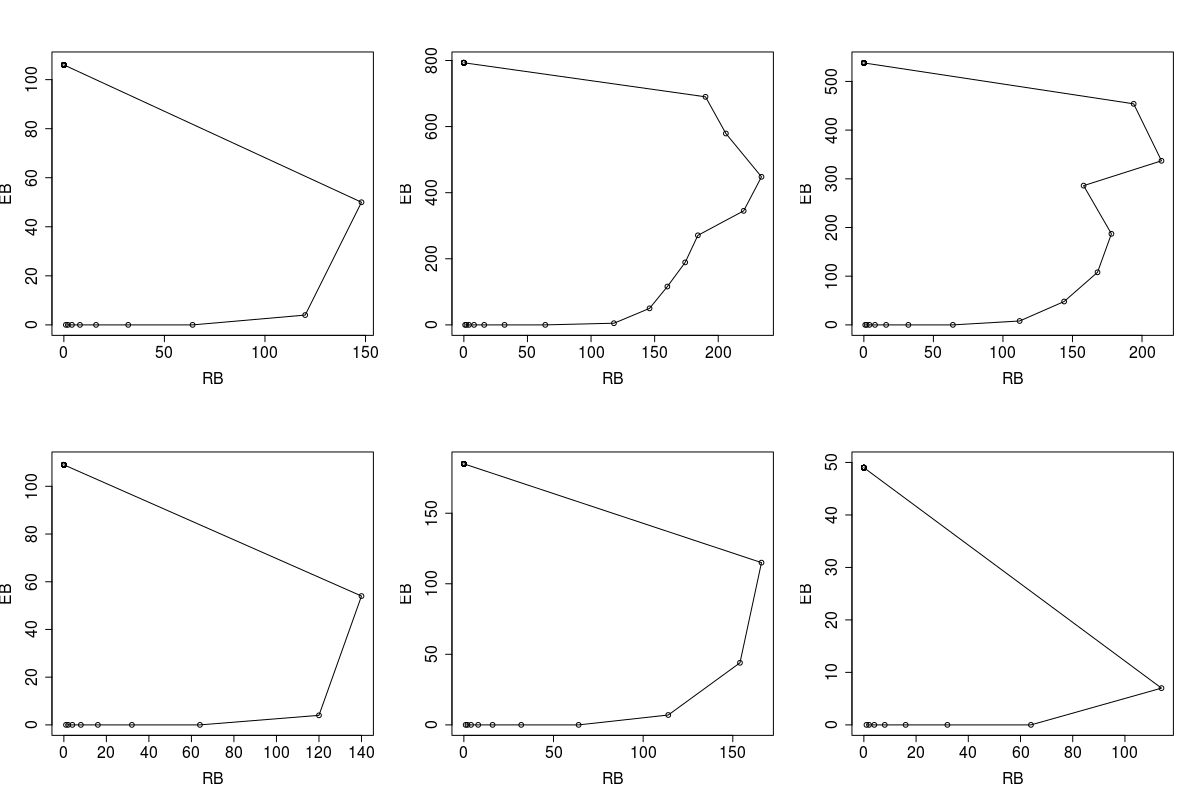}
\end{center}
\vspace*{-0.7cm}
\caption{Simulations of the process with death-rate \eqref{eq:death2},
with $(\alpha, \beta) = (3,1)$.}
\label{fig:sim-d1-3}
\end{figure}

Finally, for $(\alpha, \beta) = (3,1)$ the relative 
effect of RBs is much larger, which implies a shorter period of exponential
increase in the RB population, and a longer coexistence of RB and EB population,
see Figures \ref{fig:mean-d1-3} and \ref{fig:sim-d1-3}.
These results suggest that the effect of RBs on host cell's death is larger,
or at least as large as the effect of EBs.
Here $h(1,0) = 285.7$. The $p$ values in Figure \ref{fig:p} are even smoother than
in the previous case, and the population prefers to have not too many RBs.

\begin{figure} 
\centering
\hspace*{-0.5cm}
\begin{subfigure}{0.49\textwidth}
\includegraphics[width = \textwidth]{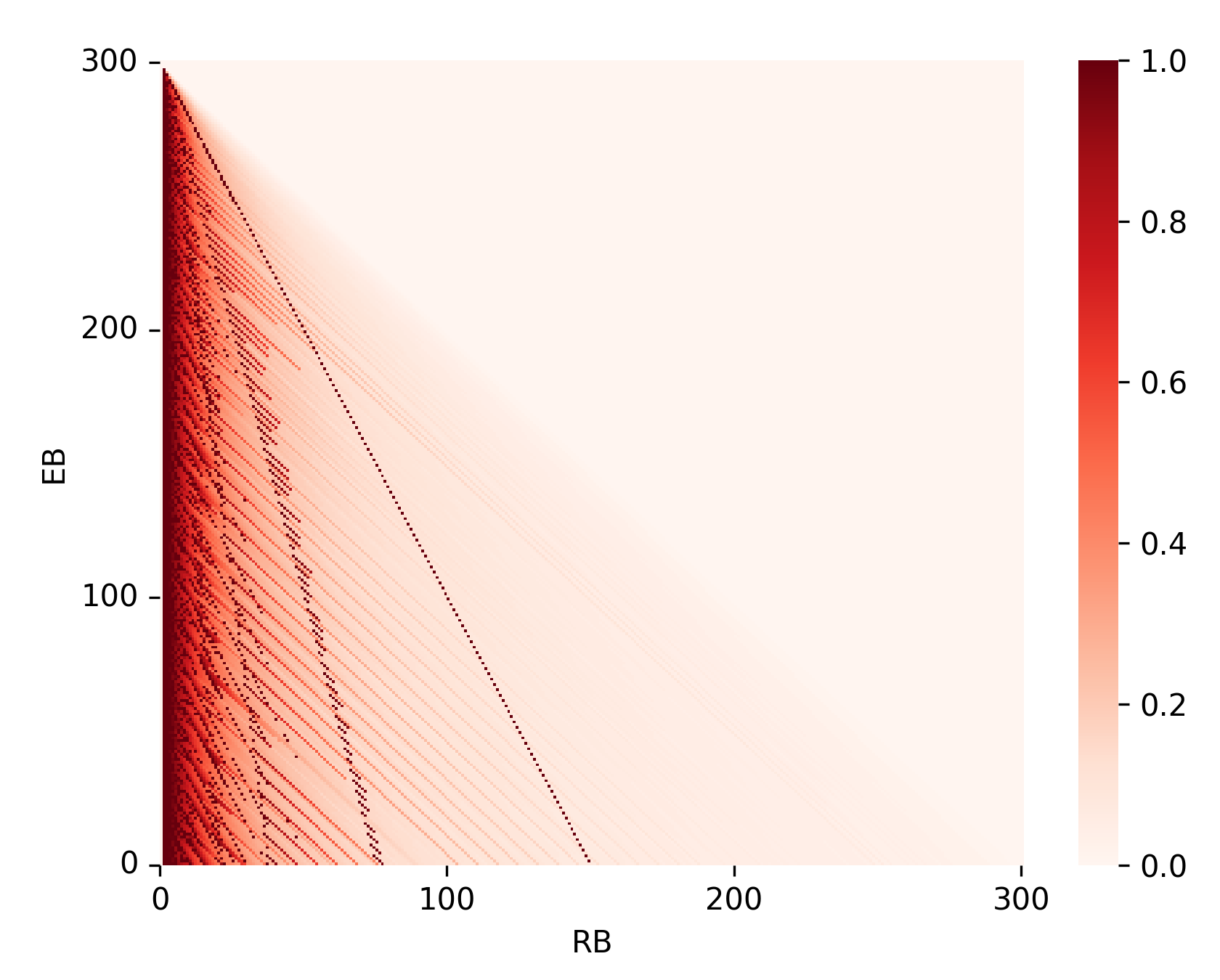} 
\end{subfigure}
\begin{subfigure}{0.49\textwidth}
\includegraphics[width = \textwidth]{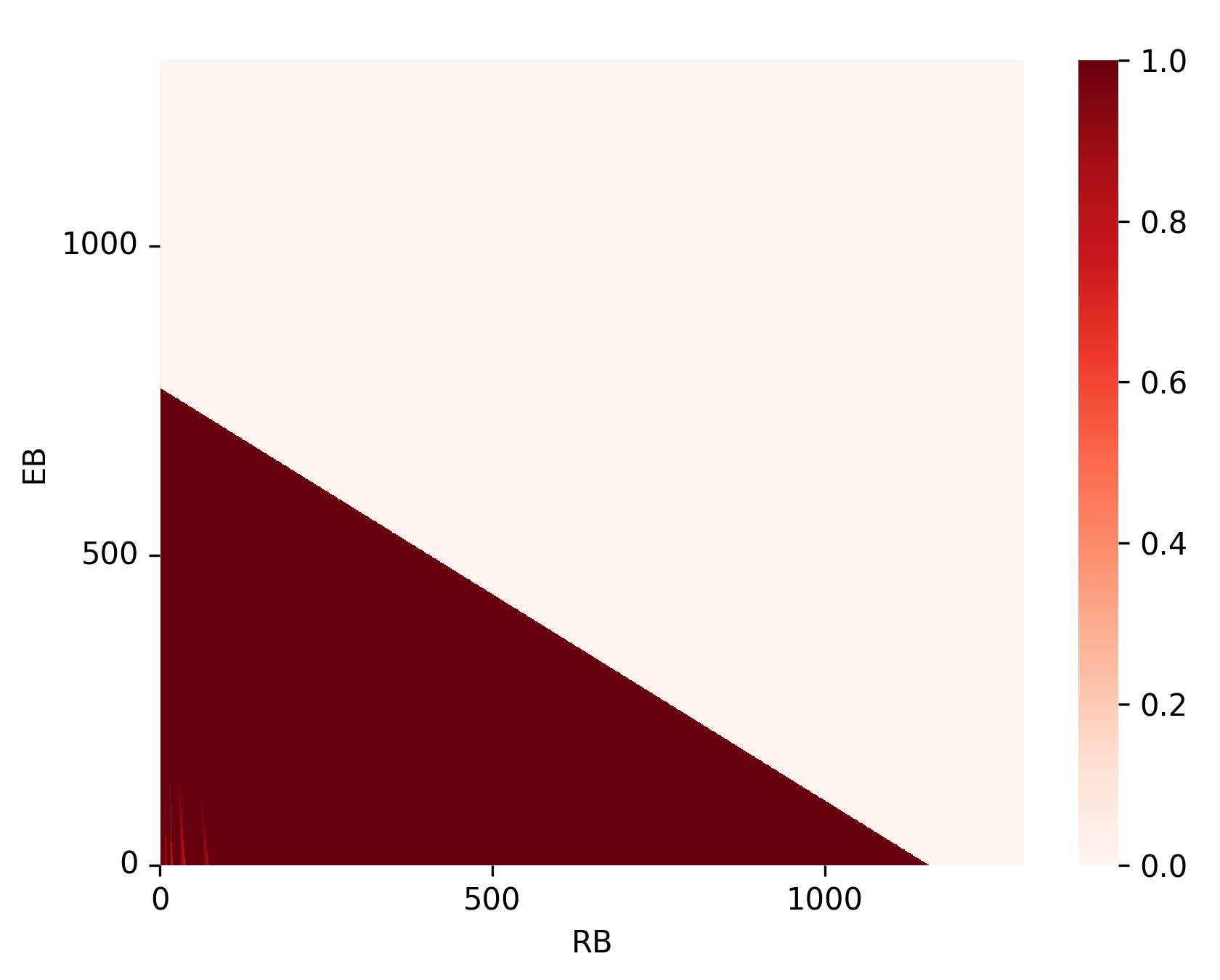} 
\end{subfigure} 
\hspace*{-0.5cm}
\begin{subfigure}{0.49\textwidth}
\includegraphics[width = \textwidth]{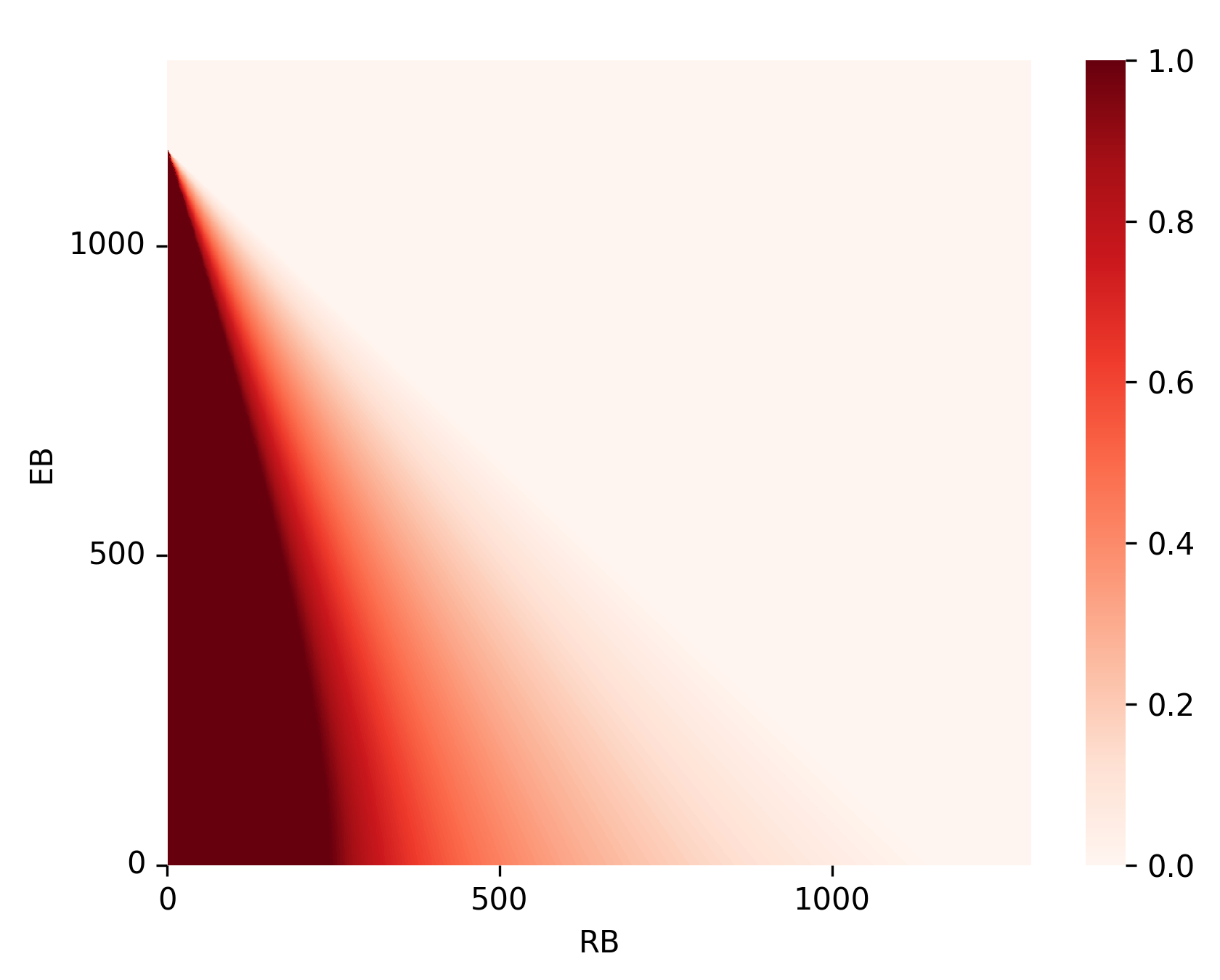} 
\end{subfigure}
\begin{subfigure}{0.49\textwidth}
\includegraphics[width = \textwidth]{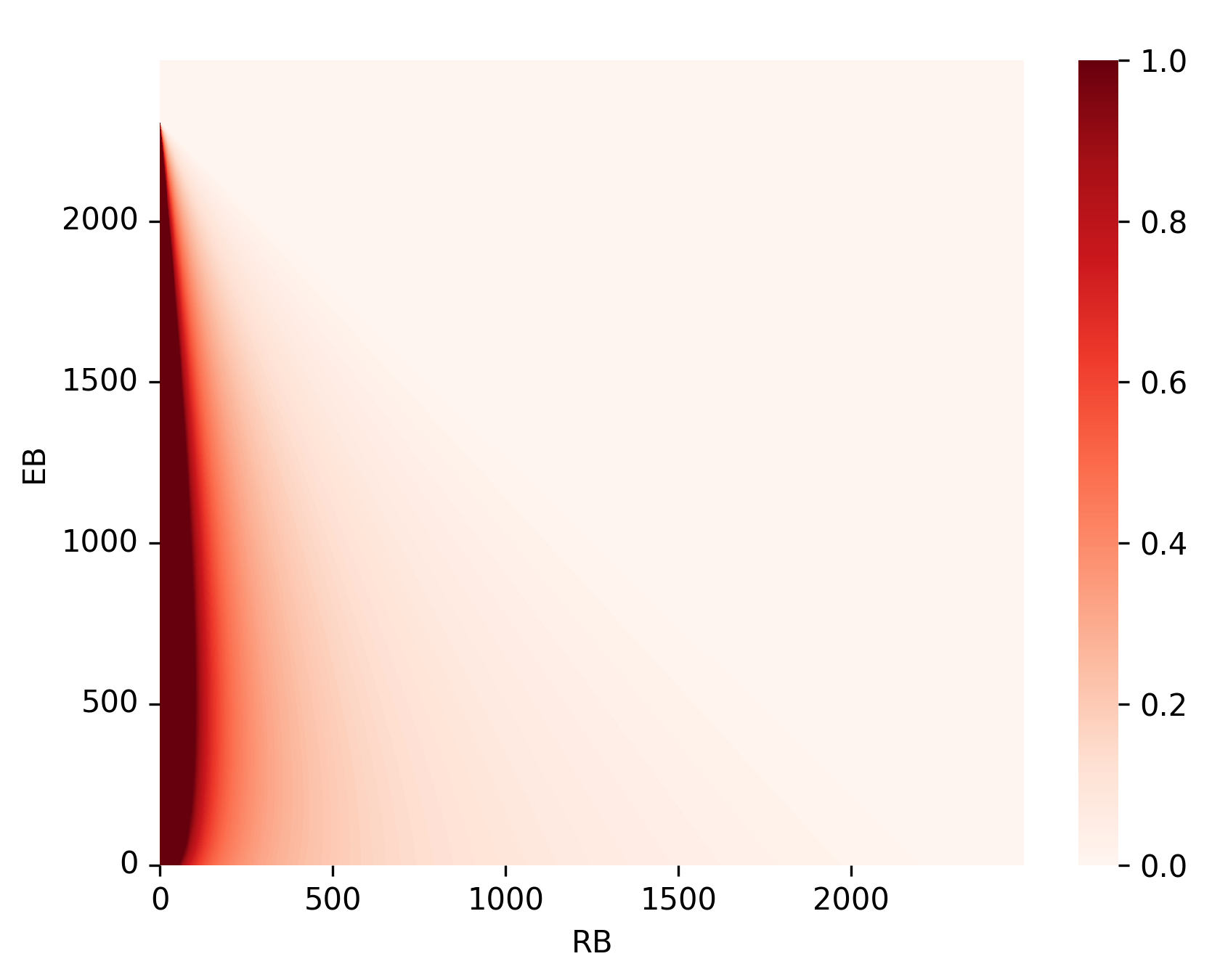} 
\end{subfigure}
\caption{The numerical $p$ values corresponding to  
the death-rate function \eqref{eq:death0} with $C = 300$ (top left), 
and to \eqref{eq:death2} with $(\alpha, \beta) = (1,3)$ (top right), 
$(\alpha, \beta) = (2,2)$
(bottom left), $(\alpha, \beta) = (3,1)$ (bottom right).}
\label{fig:p}
\end{figure}

\section{Concluding remarks}

To model the evolution of \emph{Chlamydia} populations we propose 
a novel Galton--Watson branching model, where the state-dependent 
offspring distribution is determined by solving a stochastic
optimization problem. The only input of the process is a 
death-rate function $d$, which describes the probability that 
the host cell dies in a given state.

Choosing a natural death-rate function our simulation study
shows that the process captures the asynchronous conversion
property of the population, which was recently found 
experimentally in \cite{Lee}. Moreover, our simulated data 
fits extremely well to the real measurements in \cite{Lee}, 
see Table \ref{tab:1}.
To the best of our knowledge,
this is the first mathematical model which reproduces this phenomena.

The cause of the host cell’s death is not yet well-understood.
Experiments suggests that \emph{chlamydia}
controls host cell survival, as an early death would be disadvantageous
to the bacterial population, see \cite[p.~394]{Elwell}.
However, the amount of bacteria in the host cell definitely has
a strong effect. It is not clear which form of the bacteria is 
more harmful to the host cell, since RBs are larger physically, while 
EBs secrete chemicals.
Varying the relative effect of RBs and EBs on the 
death time of the host cell, our simulation studies suggest 
that RBs and EBs have the same effect  on the host cell's death.

\bigskip

\noindent \textbf{Acknowledgements.}
We are grateful to Dezs\H{o} Virok for explaining us the 
necessary biology. PK is partially supported 
by the J\'anos Bolyai Research Scholarship of the Hungarian
Academy of Sciences. MSz is supported by the \'UNKP-22-3-SZTE-457 new national excellence 
program of the ministry for culture and innovation from the source of the national 
research, de\-ve\-lop\-ment and innovation fund.
This research was supported by the Ministry of Innovation and
Technology of Hungary from the National Research, Development
and Innovation Fund, project no.~TKP2021-NVA-09.

%\bibliographystyle{plain}
%\bibliography{modell}

\end{document}